\documentclass[10pt]{article}

\usepackage{a4wide}
\usepackage[sc]{mathpazo}
\usepackage{amsmath}
\usepackage{amssymb}
\usepackage{amsthm}
\usepackage{lmodern}
\usepackage[framemethod=tikz]{mdframed}

\usepackage{enumitem}
\usepackage{color}

\usepackage{stmaryrd}

\DeclareFontFamily{U}{mathx}{\hyphenchar\font45}
\DeclareFontShape{U}{mathx}{m}{n}{
      <5> <6> <7> <8> <9> <10>
      <10.95> <12> <14.4> <17.28> <20.74> <24.88>
      mathx10
      }{}
\DeclareSymbolFont{mathx}{U}{mathx}{m}{n}
\DeclareFontSubstitution{U}{mathx}{m}{n}
\DeclareMathAccent{\widecheck}{0}{mathx}{"71}
\DeclareMathAccent{\wideparen}{0}{mathx}{"75}

\RequirePackage{filecontents}        % loading package filecontents

\usepackage[numbers]{natbib}
\usepackage[colorlinks]{hyperref}    % better urls in bibliography
\usepackage{cleveref}
\hypersetup{colorlinks,
    linkcolor={red!50!black},
    citecolor={blue!50!black},
    urlcolor={blue!80!black}
}

\usepackage{graphicx}
\usepackage[font=small,labelfont=bf]{caption}
\usepackage{subcaption}
\usepackage{float}
\usepackage{tikz}
\usetikzlibrary{decorations.markings,calc}
\def\ladderwidth{2.8mm}
\def\ladderstep{3.8mm}
\tikzset{
  ladder/.style={decorate,decoration={
      markings,
      mark=between positions {1/#1/2} and {-1/#1/2} step {1/#1} with {
        \draw[thick,black]
        (0,-\ladderwidth) -- (0,\ladderwidth)
        (\ladderstep,-\ladderwidth) -- (-\ladderstep,-\ladderwidth)
        (\ladderstep,\ladderwidth) -- (-\ladderstep,\ladderwidth);
      }
    },
  },
  XOR/.style={draw,circle,append after command={
      [shorten >=\pgflinewidth, shorten <=\pgflinewidth,]
      (\tikzlastnode.north) edge (\tikzlastnode.south)
      (\tikzlastnode.east) edge (\tikzlastnode.west)
      }
    },
  mmt/.style = {regular polygon, regular polygon sides=3,
            draw, fill=white,
            inner sep=0.9mm, outer sep=0mm,
            shape border rotate=-90},
  bigmmt/.style = {regular polygon, regular polygon sides=3,
            draw, fill=white,
            inner sep=2mm, outer sep=0mm,
            shape border rotate=-90},
  operator/.style = {draw,fill=white,minimum size=1.5em},
  phase/.style = {fill,shape=circle,minimum size=5pt,inner sep=0pt},
  not/.style = {shape=circle,minimum size=5pt,inner sep=0pt},
  surround/.style = {fill=blue!10,thick,draw=black,rounded corners=2mm},
}
\usetikzlibrary{backgrounds,fit,shapes,decorations.pathreplacing,patterns,arrows}
\usepackage{anyfontsize}

\usepackage{algorithm}
\usepackage[noend]{algpseudocode}
\usepackage{xpatch}

\algrenewcommand\alglinenumber[1]{{\sffamily\footnotesize#1}}

\usepackage[margin=1in]{geometry}
\hypersetup{pdfpagemode=UseNone}
\usepackage{pdfsync}

\makeatletter
\newsavebox\myboxA
\newsavebox\myboxB
\newlength\mylenA
\newcommand*\xoverline[2][0.75]{%
    \sbox{\myboxA}{$\m@th#2$}%
    \setbox\myboxB\null% Phantom box
    \ht\myboxB=\ht\myboxA%
    \dp\myboxB=\dp\myboxA%
    \wd\myboxB=#1\wd\myboxA% Scale phantom
    \sbox\myboxB{$\m@th\overline{\copy\myboxB}$}%  Overlined phantom
    \setlength\mylenA{\the\wd\myboxA}%   calc width diff
    \addtolength\mylenA{-\the\wd\myboxB}%
    \ifdim\wd\myboxB<\wd\myboxA%
       \rlap{\hskip 0.5\mylenA\usebox\myboxB}{\usebox\myboxA}%
    \else
        \hskip -0.5\mylenA\rlap{\usebox\myboxA}{\hskip 0.5\mylenA\usebox\myboxB}%
    \fi}
\makeatother

\def\natural{\mathbb{N}}

\def\field{\mathbb{F}}

\def\union{\cup}

\let\emptyset\varnothing

\def\conj{^{\dagger}}
\def\to{\rightarrow}
\newcommand{\ip}[2]{\langle #1, #2 \rangle}

\newcommand{\1}{\mathbb{1}}

\renewcommand{\dim}[1]{\operatorname{dim}\left( #1 \right)}
\renewcommand{\deg}{\sansserif{deg}}

\newcommand{\comp}[1]{\check{#1}}
\def\half{\frac{1}{2}}

\def \lket {\left|}
\def \rket {\right\rangle}
\def \lbra {\left\langle}
\def \rbra {\right|}
\newcommand{\ket}[1]{\lket\mspace{0.5mu} #1 \mspace{0.5mu}\rket}
\newcommand{\encket}[1]{\xoverline{\ket{#1}}}
\newcommand{\bra}[1]{\lbra\mspace{0.5mu} #1 \mspace{0.5mu}\rbra}
\newcommand{\encbra}[1]{\xoverline{\bra{#1}}}

\newcommand{\kb}[1]{\ket{#1}\bra{#1}}

\def\css{\sansserif{CSS}}

\def\A{\mathcal{A}}
\def\B{\mathcal{B}}
\def\C{\mathcal{C}}

\def\E{\mathcal{E}}

\def\G{\mathcal{G}}

\def\K{\mathcal{K}}

\def\M{\mathcal{M}}

\def\P{\mathcal{P}}

\def\R{\mathcal{R}}

\def\T{\mathcal{T}}

\def\V{\mathcal{V}}

\def\X{\mathcal{X}}
\def\Y{\mathcal{Y}}
\def\Z{\mathcal{Z}}

\def\a{\operatorname{A}}

\def\g{\operatorname{G}}
\def\i{\operatorname{I}}
\def\j{\operatorname{J}}
\def\h{\operatorname{H}}

\def\tC{\widetilde{\C}}

\def\th{\widetilde{\operatorname{H}}}

\def\uhat{\hat{u}}

\def\Pr{\operatorname{Pr}}
\def\avg{\mathbb{E}}

\newtheorem{theorem}{Theorem}
\newtheorem*{theorem*}{Theorem}
\newtheorem{lemma}[theorem]{Lemma}

\newtheorem{corollary}[theorem]{Corollary}

\newtheorem{definition}[theorem]{Definition}

\usepackage{thmtools}
\usepackage{thm-restate}

\def\path{\operatorname{P}}

\newcommand{\sansserif}[1]{%
  \ifmmode
  \mathsf{#1}%
  \else
   \textsf{#1}%
  \fi
}

\def\ev{\sansserif{ev}}
\def\sfq{\sansserif{Q}}

\def\sfu{\sansserif{U}}
\def\sfut{\widetilde{\sansserif{U}}}
\def\sfw{\sansserif{W}}
\def\sfv{\sansserif{V}}

\def\cz{\sansserif{cZ}}

\def\cleq{\preccurlyeq}

\renewcommand{\int}[2]{\Delta_{#1} \G'_{#2}}

\def\sfu{\sansserif{U}}

\newcommand{\pr}[1]{\sansserif{Pr}\left( #1 \right)}
\usepackage{authblk}

\title{Magic state distillation with punctured polar codes}
\author[1]{Anirudh Krishna}
\author[2]{Jean-Pierre Tillich}
\affil[1]{
  Universit\'e de Sherbrooke\\
  2500 Boulevard de l'Universit\'e\\
  Sherbrooke, QC J1K 2R1, Canada
  }
\affil[2]{
  Inria, Team SECRET\\
  2 Rue Simone IFF, CS 42112\\
  75589 Paris Cedex 12, France
  }
\date{}

\begin{document}
\maketitle
\begin{abstract}
  We present a scheme for magic state distillation using punctured polar codes.
  Our results build on some recent work by Bardet et al. \cite{bardet2016algebraic} who discovered that polar codes can be described algebraically as decreasing monomial codes.
  Using this powerful framework, we construct tri-orthogonal codes \cite{bravyi2012magic} that can be used to distill magic states for the $T$ gate.
  An advantage of these codes is that they permit the use of the successive cancellation decoder whose time complexity scales as $O(N\log(N))$.
  We supplement this with numerical simulations for the erasure channel and dephasing channel.
  We obtain estimates for the dimensions and error rates for the resulting codes for block sizes up to $2^{20}$ for the erasure channel and $2^{16}$ for the dephasing channel.
  The dimension of the triply-even codes we obtain is shown to scale like $O(N^{0.8})$ for the binary erasure channel at noise rate $0.01$ and $O(N^{0.84})$ for the dephasing channel at noise rate $0.001$.
  The corresponding bit error rates drop to roughly $8\times10^{-28}$ for the erasure channel and $7 \times 10^{-15}$ for the dephasing channel respectively.
\end{abstract}
\section{Introduction}

Implementations of quantum circuits are imperfect and prone to error.
In order to realize scalable quantum computers, we need to construct quantum circuits capable of working with unreliable components.
This is the focus of the domain of fault-tolerant quantum computation \cite{shor1996fault,aharonov1997fault,kitaev1997quantum,knill1998resilient}.
The premise behind this theory is to encode quantum information using \emph{quantum error correcting codes} which serve as a buffer against noise.

To process encoded information, we require some way to perform logical operations without unencoding it and thereby leaving it vulnerable to errors.
Of particular interest is the technique called state injection, a scheme to inject special ancilla states called magic states into a quantum circuit \cite{bravyi2005universal}.
These states must undergo a resource intensive purification process called magic state distillation to ensure that they have high fidelity.
Current estimates state that a large fraction of qubits required for quantum computation will have to be dedicated to this process \cite{jones2013multilevel,o2017quantum}.
It is therefore imperative to reduce the overhead required by magic state distillation as this is a bottleneck.

We shall focus on magic state distillation protocols to distill the state $\ket{A} = (\ket{0} + e^{i\pi/4}\ket{1})/\sqrt{2}$.
This state can be used to inject the $T$ gate, where 
\begin{align}
  \label{eq:defT}
  T = 
  \begin{pmatrix}
    1 & 0\\
    0 & e^{i\pi/4}
  \end{pmatrix}~.
\end{align}
Together with gates from the Clifford group, this gate can be used to implement a universal set of gates.
There are several magic state distillation protocols to distill the $\ket{A}$ state such as \cite{bravyi2005universal,hastings2017distillation,haah2017codes} that use quantum error correcting codes such as the Reed-Muller code with special symmetry properties.
These symmetries imply that applying transversal, physical $T$ gates to the physical qubits of the quantum error correcting code implements the logical $T$ gate on the encoded qubits.
Decoding the quantum error correcting code then results in the magic state $\ket{A}$.
The framework of tri-orthogonality defined by Bravyi and Haah \cite{bravyi2012magic} extends these symmetries to describe entire code families.
This work also defined a protocol to distill states using tri-orthogonal codes.
One key feature of the Bravyi-Haah protocol is that it relies on post-selection, i.e. the states at the end of the protocol are accepted if and only if a certain measurement outcome is obtained; if not, the states are discarded.
The probability of obtaining this measurement outcome decreases as the distance of the quantum error correcting code increases.

In this paper, we present a technique to distill the state $\ket{A}$ using punctured polar codes.
We shall show that these codes are tri-orthogonal and can encode a growing number of logical qubits similar to \cite{hastings2017distillation}.
Rather than use post selection, we modify the Bravyi-Haah protocol and use Steane error correction to correct errors and decode our code.
This implies a trade-off between the rate of distillation and the error rate the scheme can tolerate.

Discovered by Arikan \cite{arikan2009channel}, polar codes were the first family of codes that were shown to efficiently achieve the capacity of a binary discrete memoryless channel (B-DMC).
These codes come equipped with an iterative decoder called the successive cancellation (SC) decoder whose decoding complexity scales as $O(N\log(N))$ for a code of block-size $N$.
We build on some recent results by Bardet et al. wherein polar codes are cast as \emph{decreasing monomial codes} \cite{bardet2016algebraic}.
These algebraic tools are central to demonstrating the existence of tri-orthogonal polar codes.
We hope these tools will be useful in the study of quantum error correcting codes in general.

Notably, our magic state distillation has a low decoding complexity as the quantum codes inherit the SC decoder whose time complexity is $O(N\log(N))$.
We supplement our analysis with some numerical results which indicate upper bounds on the size of the resulting codes and their error rates for block sizes of interest.

Although Reed-Muller codes are closely related to polar codes, the SC decoder for polar codes has a much better error correction capacity than the majority logic decoding algorithm for Reed-Muller codes \cite{dumer2017recursive,arikan2010survey}.
Successive and iterative decoders exist for the Reed-Muller code as well, but again correct far fewer errors than their polar code counterparts \cite{arikan2008performance}.
For this reason, using the distance of the polar code to characterize its error correction capacity can be misleading.

\textbf{Related work:}
Polar codes have been generalized to the quantum realm \cite{renes2012efficient, renes2014polar, renes2015efficient, wilde2012quantum, wilde2013CQ, wilde2013degradable, wilde2013towards} but were studied from the perspective of quantum Shannon theory.
These authors were not concerned with fault tolerance and thus did not explore transversal gates on quantum codes.

There also exist schemes to distill states to implement gates besides the $T$ gate \cite{reichardt2005quantum,meier2012magic,haah2017magic1,haah2017magic2} and also complex schemes that combine both distillation and compilation such as \cite{landahl2013complex,duclos2015reducing,campbell2017unifying,campbell2017unified}.

Finally, magic state distillation can be performed using qudits ($d$ dimensional quantum systems) rather than qubits \cite{anwar2012qutrit,campbell2012magic,campbell2014enhanced,krishna2018towards}.
These schemes have the potential to reduce the overhead associated with magic state disillation.

\textbf{Outline:} This paper is structured as follows.
In section \ref{sec:background}, we define binary polar codes, and proceed to describe the algebraic formalism discovered by Bardet et al. \cite{bardet2016algebraic}.
For the sake of completeness, we review the tri-orthogonality condition and its application to magic states.
In section \ref{sec:distillationprocedure}, we discuss our distillation procedure which assumes that the error correcting code is subject to erasure noise or dephasing noise.
We present our results in section \ref{sec:triorthogonalcodesfrompolarcodes}.
First, we compute the dimension of the triply-even codes constructed from polar codes and plot this parameter as a function of the block-size in \ref{subsec:dim}.
Second, we show that it is possible to construct tri-orthogonal codes from polar codes that achieve good performance for erasure channels and for dephasing channels.

\section{Background}
\label{sec:background}
\subsection{Polar codes}
\label{subsec:polarcodes}
We begin by briefly reviewing the theory of polar codes.
Throughout the paper, we let $n \in \natural$ denote a natural number and $N = 2^{n}$.
Suppose we wish to transmit a message $x \in \field_2^N$ across $N$ copies of a B-DMC $\sfw: \X \to \Y$, where $\X = \field_2$ and $\Y$ is some output alphabet.
The fundamental problem of coding theory is to deduce the input word $x$ having received a corrupted word $y \in \Y^N$.
Arikan's polarization technique reduces the task of the decoder to an iterative process.

For a given channel $\sfw$, block size $N$ and error rate $\epsilon \in [0,1]$, the binary polar code $\C = \C(N,\epsilon)$ is specified by a set $\A \subseteq \{1,...,N\}$, where $|\A| = K$ is the dimension of $\C$ and $\epsilon \in [0,1]$ is an upper bound on the bit error rate of the code.
We shall discuss how to obtain the set $\A$ shortly and begin by discussing the encoding and decoding process assuming this set has been provided.
The information we wish to transmit is stored in a vector $u \in \field_2^N$ where for $a \in \A$, the indices $u_a$ carry information and the other indices are said to be \emph{frozen}, i.e. for $b \in \A^c$, $u_b = 0$.
At the core of the encoding process is the $2 \times 2$ matrix $F$ defined as
\begin{align}
  F =
  \begin{pmatrix}
    1 & 1\\
    0 & 1
  \end{pmatrix}~.
\end{align}
Given an input word $u \in \field_2^{N}$, the encoder maps it to $x \in \field_2^N$, where $x = F^{\otimes n} u$.
This word $x$ is then transmitted across $N$ copies of the noisy channel $\sfw$ resulting in the corrupted word $y$.

The decoder is provided the set $\A$ ahead of time.
Given $y$, it deduces the bits of $u$ sequentially (from $1$ to $N$) treating the bits that it has not yet decoded as noise.
At the $a$-th iteration, it sees the synthetic channel $\sfw^{(a)}$ where
\begin{align*}
  \sfw^{(a)}(y,u_1,\cdots u_{a-1}|u_{a}) = \frac{1}{2^N}\sum_{u_{a+1},...,u_{N}}\sfw(y_1|x_1)\cdots\sfw(y_N|x_N)~.
\end{align*}
It uses this formula to compute the log likelihood ratios $\lambda_a$ where
\begin{align*}
  \lambda_a = \log_2 \frac{\sfw^{(a)}(y,u_1,\cdots,u_{a-1}|0)}{\sfw^{(a)}(y,u_1,\cdots,u_{a-1}|1)}~.
\end{align*}
Having computed the $a$-th log-likelihood ratio, it estimates the $a$-th bit as
\begin{align*}
  u_a =
  \begin{cases}
    0 &\mbox{if } \lambda_a > 0 \mbox{ or } a \in \A^c\\
    1 &\mbox{if } \lambda_a < 0\\
  \end{cases}~.
\end{align*}
Although it is not evident from this presentation, this computation can be performed in $O(N\log(N))$ steps.

It can be shown that the probability of error of the synthetic channels can be upper bounded using the \emph{Bhattacharyya parameter} $\B(\sfw)$ defined for a B-DMC $\sfw: \field_2 \to \Y$ as
\begin{align}
  \B(\sfw) := \half \sum_{y \in \Y} \sqrt{\sfw(y|0)\sfw(y|1)}~.
\end{align}
The Bhattacharyya parameter obeys $0 \leq \B(\sfw) \leq 1$, with $\B(\sfw) = 0$ indicating a noiseless channel and $\B(\sfw) = 1$ indicating a totally unreliable channel.

Exact code construction for a polar code $\C = \C(N,\epsilon)$ whose bit error rate is upper bounded by $\epsilon \in [0,1]$ employs the Bhattacharyya parameters to deduce the set $\A$.
For each index $i \in \{1,...,N\}$, we can evaluate $\B(\sfw^{(a)})$ by transmitting the all zero codeword $0^N$ across the channel $\sfw^{N}$.
Given a threshold $\epsilon \in [0,1)$, we can construct the polar code $\C$ by choosing those indices $a$ such that $\B(\sfw^{(a)}) \leq \epsilon$ to constitute the set $\A$.

Two channels we shall be interested in are the erasure channel and the binary symmetric channel.
When a symbol is erased, we replace it with a special symbol $\perp$.
For ease of representation, we denote a probability distribution $\{\pr{0},\pr{1},\pr{\perp}\}$ over symbols in $\field_2 \union \{\perp\}$ as $\pr{0}[0] + \pr{1}[1] + \pr{\perp}[\perp]$.

For $p \in [0,1]$, the action of the erasure channel, henceforth denoted $\sfw_p$, on $x \in \field_2$ is then
\begin{align*}
  \sfw_p(x) = (1-p)[x] + p[\perp]~,
\end{align*}
which is equivalent to saying that the conditional probability of obtaining $x$ at the output given that $x$ was transmitted is $1-p$ and the probability of seeing $\perp$ is $p$.
Similarly, for $p \in [0,1]$, the action of the binary symmetric channel, henceforth denoted $\sfv_p$, on $x \in \field_2$ is then
\begin{align*}
  \sfv_p(x) = (1-p)[x] + p[\bar{x}]~,
\end{align*}
where $\bar{x} = x+1 \pmod{2}$.

\subsection{Decreasing monomial codes}
\label{subsec:decreasingmonomialcodes}
The drawback with the above construction is that it is not universal; the frozen channels have to be recomputed for each B-DMC $\sfw$.
Although optimizing the polar code this way permits us to provably achieve the capacity efficiently, it does not provide any insight into which synthetic channels are used to transmit information and which synthetic channels are frozen.
This is also a problem in practice because exact code construction is difficult for most channels.
Devising good algorithms for code construction is non-trivial and has been the subject of considerable research \cite{mori2009performance}, \cite{pedarsani2011construction}, \cite{tal2013construct}.

Recently, Bardet et al. \cite{bardet2016algebraic} discovered that there exists an algebraic framework describing polar codes which sheds some light on which channels carry information and which channels are frozen.
We shall exploit this algebraic structure in order to construct tri-orthogonal quantum codes in the following section.
Here we shall briefly review relevant concepts from \cite{bardet2016algebraic} required to discuss magic state distillation.
For a complete development, including proofs of claims in this section, we refer the reader to the original paper by Bardet et al. \cite{bardet2016algebraic}.

Let $n \in \natural$ be some natural number and $\R_n$ be the ring of polynomials over $n$ variables defined as
\[
  \R_n = \field_2[x_0,...,x_{n-1}]/(x^2_0 - x_0,...,x^2_{n-1}-x_{n-1})~.
\]

Any monomial $\alpha \in \R_n$ can be expressed as $x_0^{a_0}...x_{n-1}^{a_{n-1}}$ for some exponents $a_0,...,a_{n-1} \in \field_2$.
The set of all monomials is denoted $\M_n$ where
\begin{align*}
  \M_n = \{x_0^{a_0}...x_{n}^{a_{n-1}}| a_0,...,a_{n-1} \in \field_2\}~.
\end{align*}
Equivalently, we may directly identify the monomial $\alpha$ with the exponents of $x_i$ as $(a_0,...,a_{n-1})$.
The degree of the monomial $\alpha$ is the number of non-zero exponents and is denoted $\deg(\alpha)$.

We can also uniquely identify $\alpha$ with a vector over $\field_2^{N}$ by evaluating it over the field $\field_2^{n}$.
Define the evaluation map $\ev:\R_n \to \field_{2}^{N}$ where
\begin{align}
  \ev(\alpha) = \{\alpha(u)\}_{u \in \field_{2}^{n} }~,
\end{align}
which denotes the vector each of whose components is obtained by evaluating $\alpha$ on each of the $N$ points in $\field_{2}^{n}$.
\begin{definition}[\textbf{Monomial code}]
  Let $\i \subseteq \M_{n}$ be a set of monomials in $n$ variables.
  The corresponding monomial code $\C(\i) \subseteq \field_{2}^{N}$ is defined as
  \begin{align}
    \C(\i) := \operatorname{span}\{\ev(\alpha)|\alpha \in \i \}~.
  \end{align}
\end{definition}

\begin{lemma}
  For $\i \subseteq \M_n$, the dimension of the code $\C(\i)$ is equal to $|\i|$.
\end{lemma}
\begin{proof}
  Monomials are linearly independent over $\R_n$ and the map $\ev$, being injective, preserves this relation.
\end{proof}

There is a natural partial order on $\M_n$ called the divisibility order defined as follows.
For two monomials $\alpha,\beta \in \M_n$ with exponents $\{a_i\}_i$ and $\{b_i\}_i$, we write $\alpha \cleq_{w} \beta$ if and only if for all $i \in \{0,...,n-1\}$,
\begin{align}
  a_i \leq b_i~.
\end{align}
The subscript $w$ indicates a weak order and we can define a strong order as follows.
\begin{definition}[\textbf{Strong order}]
Let $\alpha,\beta \in \M_n$ be two monomials of the same degree $r$ such that $f=x_{i_1}\cdots x_{i_r}$ and $g = x_{j_{1}}\cdots x_{j_r}$.
We say that $f \cleq g$ if and only if for all $1 \leq t \leq r$:
\begin{align}
  i_t \leq j_t~.
\end{align}
We can then extend this definition to include monomials $\alpha$ and $\beta$ of different degrees via divisibility.
If $\deg(\alpha) < \deg(\beta)$, then $\alpha \cleq \beta$ if and only if there exists a monomial $\beta^{*}$ such that $\deg(\alpha) = \deg(\beta^{*})$ and $\alpha \cleq \beta^{*} \cleq_{w} \beta$.
\end{definition}

This ordering can apply to an entire set of monomials as follows.
\begin{definition}[\textbf{Decreasing sets}]
\label{def:decreasingSet}
    A set $\i \subset \M_n$ is {\em decreasing} if and only if ($ g \in \i$ and $f \cleq g$) implies that $f \in \i$.

    Furthermore, a code $\C(\i)$ is a decreasing monomial code if $\i$ is a decreasing set.
\end{definition}

This ordering is important because of the following relation.
Given the exponent $(a_0,...,a_{n-1})$ of a monomial $\alpha$, let $a = \sum_{t = 0}^{n-1} \a_t 2^t$ be the integer corresponding to a natural ordering of the exponent.
\begin{lemma}
  \label{lem:ordering}
  Let $\alpha,\beta \in \M_n$ and $a$, $b$ be the integers corresponding to the exponents of $\alpha$, $\beta$ respectively.
  We have
  \[
    \alpha \cleq \beta \implies \B(\sfw^{(a)}) \leq \B(\sfw^{(b)})
  \]
\end{lemma}
This in turn implies the central result of Bardet et al. \cite{bardet2016algebraic} is the following theorem that we state without proof (See theorem 1 of \cite{bardet2016algebraic}).
\begin{theorem}
  Polar codes are decreasing monomial codes.
\end{theorem}
Recall that a polar code is constructed by ordering the Bhattacharyya parameters associated with the channels $\sfw^{(a)}$ for $a \in \{1,...,N\}$.
If we choose to include codewords $\ev(\beta)$ for $\beta \in \M_n$ such that $\B(\sfw^{(b)}) \leq \epsilon$, then we must necessarily have all $\alpha \cleq \beta$ because of lemma \ref{lem:ordering}.
This result is important because it allows for a simple description of the dual of a monomial code.

We denote the complement of a monomial $\alpha \in \M_{n}$ by $\comp{\alpha}$ as
\begin{align}
  \comp{\alpha}(x) = x_{0}^{a_{0} \oplus 1} ... x_{n-1}^{a_{n-1} \oplus 1}~,
\end{align}
where $\oplus$ denotes XOR.
By extension, for any subset $\i \subseteq \M_n$, the set $\comp{\i} = \{\comp{\alpha}|\alpha \in \i\}$.
Equivalently, the complement $\comp{\alpha}$ of $\alpha$ is the smallest monomial such that $\alpha\comp{\alpha} = x_0...x_{n-1}$, where $x_0...x_{n-1}$ is the complete monomial with all the $n$ variables.

We denote the product of two monomials $\alpha,\beta \in \M_n$ as $\alpha \cdot \beta$ where $(\alpha\cdot\beta)(u)= \alpha(u)\beta(u)$ for $u \in \field_2^{n}$.
For two vectors $v,w \in \field_{2}^{N}$, we denote the element-wise product as $v*w = (v_1 w_1,...,v_{N}w_{N})$.

The evaluation map $\ev$ maps the element-wise product to the product between monomials, i.e.
\begin{align}
  \ev(\alpha)*\ev(\beta) = \ev(\alpha \cdot \beta)~.
\end{align}

The following lemma is proved in \cite{bardet2016algebraic} (See Proposition 6).
\begin{lemma}
  \label{lem:dualDecreasingMonomial}
  Let $\i \in \M_n$ be a decreasing set of monomials and $\C(\i)$ be the corresponding monomial code.
  The dual of $\C(\i)$ is
  \begin{align}
    \C(\i)^{\perp} = \C(\M_n \setminus \comp{\i}) ~,
  \end{align}
  which is also a decreasing monomial code.
\end{lemma}
\begin{proof}
  For $\alpha,\beta \in \M_n$, we may write the inner product $\ip{\ev(\alpha)}{\ev(\beta)}$ as
  \begin{align}
    \label{eq:sum}
    \ip{\ev(\alpha)}{\ev(\beta)} = \sum_{u \in \field_2^{n}} \alpha(u)\beta(u)~.
  \end{align}
  Note that by symmetry, the only monomial that is non-zero when summed over $\field_{2}^{n}$ is the complete monomial $x_0...x_{n-1}$.
  Therefore the sum in eq. (\ref{eq:sum}) is non-zero if and only if $\comp{\alpha} \cleq \beta$.
  Therefore if we want the inner product to be $0$ for all $\beta \in \i$, then we require $\alpha \in \comp{\i}$.
  This proves the claim that $\C(\i)^{\perp} = \C(\M_n \setminus \comp{\i})$.
  
  To prove that $\C(\M_{n} \setminus \comp{\i})$ is also a decreasing monomial code, we can show that $\M_{n} \setminus \comp{\i}$ is a decreasing set.
  Let $\beta \in \M_{n} \setminus \comp{\i}$ and $\alpha \in \M_{n}$ such that $\alpha \cleq \beta$.
  For the sake of contradiction, assume that $\alpha \notin \M_{n} \setminus \comp{\i}$.
  This in turn means that $\comp{\alpha} \in \i$.
  However, if $\alpha \cleq \beta$, it implies that $\comp{\beta} \cleq \comp{\alpha}$.
  Since $\i$ is a decreasing set \ref{def:decreasingSet}, this means that $\comp{\beta} \in \i$ or that $\beta \in \comp{\i}$, which is a contradiction.
  
  This establishes the result.
\end{proof}

\subsection{Tri-orthogonal quantum codes}
\label{subsec:triorthogonalcodes}
For the sake of completeness, we review the definition of a triply-even space and a tri-orthogonal code as in \cite{bravyi2012magic}, \cite{haah2017codes}.
These notions are key to understanding quantum error correcting codes which promote physical transversal $T$ gates to logical transversal $T$ gates.

We denote by $X$ and $Z$ the Pauli operators
\begin{align*}
  X =
  \begin{pmatrix}
    0 & 1\\
    1 & 0
  \end{pmatrix}~,\qquad
  Z = 
  \begin{pmatrix}
   1 & 0\\
   0 & -1 
  \end{pmatrix}~.
\end{align*}
Together with the phase $i\1$, these operators generate the Pauli group $\P = \langle i\1, X, Z \rangle$ and can be extended to $N$ qubits as $\P_N = \P^{\otimes N}$.
For $a,b \in \field_2^N$, we let $X(a) = \bigotimes_{i=1}^{N} X^{a_i}$ and $Z(b) = \bigotimes_{j=1}^{N} Z^{b_j}$.
The Pauli group forms the first level of the Clifford hierarchy, $\K^{(1)}$.
The Clifford group $\K^{(2)}$ is defined as the set of automorphisms of the Pauli group, i.e.
\begin{align*}
  \K^{(2)} = \{U| \forall P \in \P_n: UPU\conj \in \P_N \}~.
\end{align*}
This group can be generated by the phase-gate $S = \sqrt{Z}$, the Hadamard gate $H = (X+Z)/\sqrt{2}$ and the controlled-$Z$ gate $\cz = \kb{0} \otimes \1 + \kb{1} \otimes Z$.
By themselves, these gates are insufficient to generate a universal gate set.
This is rectified by using gates from the third level of the Clifford hierarchy $\K^{(3)}$,
\begin{align*}
  \K^{(3)} = \{U| \forall P \in \P_n: UPU\conj \in \K^{(2)} \}~.
\end{align*}
Any gate from the third level of the Clifford hierarchy, together with the Clifford group $\K^{(2)}$ is sufficient to generate a universal gate set.
One such gate is the $T$ gate, where $T = \operatorname{diag}(1,e^{i\pi/4})$.

CSS codes are a class of quantum error correcting codes for which all the stabilizer generators are either composed entirely of $X$ operators and identity $I$ or entirely of $Z$ operators and identity $I$ \cite{calderbank1996good,steane1996multiple}.
This can be defined using two codes $\C_X,\C_Z \subset \field_2^{N}$ such that $\C_Z^{\perp} \subset \C_X$.
We obtain a code by mapping the spaces $\C_X^{\perp}$ and $\C_Z^{\perp}$ to stabilizer generators of $X$ and $Z$ type respectively, i.e. for $a \in \C_X^{\perp}$ and $b \in \C_Z^{\perp}$, we define stabilizer generators $X(a)$ and $Z(b)$.
We seek CSS codes for which the $T$ gate applied transversally on the physical level is equivalent to a $T$ gate applied to the logical qubits encoded by the code.
To this end, we consider code spaces with the following properties.

\begin{definition}[\textbf{Triply even space}]
  A subspace $\V \subseteq \field^{n}$ is said to be \emph{triply-even} if for any triple $u,v,w \in \V$,
  \begin{align*}
    |u * v * w| = 0 \pmod{2}~.
  \end{align*}
\end{definition}

\begin{definition}[\textbf{Tri-orthogonal matrix}]
Let $\h \in \field_2^{m \times n}$ be a matrix whose rows are labelled $\{h^{(a)}\}_{a=1}^{m}$.
We say that $\h$ is a tri-orthogonal matrix if and only if
\begin{enumerate}
  \item for $1 \leq a < b \leq m$
        \begin{align*}
          |h^{(a)}*h^{(b)}| = 0\pmod{2}
        \end{align*}
  \item for $1 \leq a < b < c \leq m$,
        \begin{align*}
          |h^{(a)}*h^{(b)}*h^{(c)}| = 0\pmod{2}
        \end{align*}
\end{enumerate}
\end{definition}

We shall partition $\h$ into two matrices $\h_0$ and $\h_1$, where $\h_0$ contains the $(m-k)$ even weight rows of $\h$ and $\h_1$ contains the $k$ odd weight rows of $\h$ for some $k \in \natural$.
Let $\g$ be the matrix whose rows are orthogonal to $\h$.

To obtain a tri-orthogonal quantum code $\css(X,\h_0; Z,\g)$ from the tri-orthogonal matrix $\h$, we associate 
\begin{enumerate}
  \item the rows of $\h_0$ with the $X$ stabilizer generators;
  \item the rows of $\g$ with the $Z$ stabilizer generators;
  \item the rows of $\h_1$ to both the $X$ and $Z$ logical operators.
\end{enumerate}

As shown in the paper by Bravyi and Haah, such a code is useful because they promote transversal $T$ gates to logical $T$ gates.
Letting $T_n = T^{\otimes n}$, lemma 2 from \cite{bravyi2012magic} states that
\begin{lemma}
\label{lem:bravyihaah}
Suppose $\css(X,\h_0; Z,\g)$ is a tri-orthogonal code.
Then there exists a diagonal unitary operator $U$ in the Clifford group $\K^{(2)}$ composed only of $\cz$ and $S$ such that
  \begin{align}
    \encket{A^{\otimes k}} = UT_{n}\encket{+^{\otimes k}}~.
  \end{align}
\end{lemma}

\section{Distillation procedure}
\label{sec:distillationprocedure}

In this section, we describe the procedure we shall use to distill magic states.
Let $\tC = [N,K]$ be a binary linear code defined by a parity check matrix $\th \in \field_2^{(N-K)\times N}$.
If the dual code $\tC^{\perp}$ is triply-even, we may construct a quantum code by puncturing the columns of $\th$.
For simplicity, suppose we puncture the first $k$ bits of the space; if not, we can always permute the bits.
We may then use Gaussian elimination to express the parity check matrix $\th$ of the code $\tC$ in systematic form with respect to the puncture:
\begin{align}
  \th =
  \begin{pmatrix}
    \1_k & \h_{1}\\
    0  & \h_{0}
  \end{pmatrix} ~,
\end{align}
where $\1_k$ represents the $k \times k$ identity matrix, and $\h_{1}$, $\h_{0}$ are matrices of dimension $k \times (N-k)$ and $(K-k) \times (N-k)$ respectively.
The matrices $\h_1$ and $\h_0$ have rows whose weights are $1$ and $0$ mod $2$ respectively.
Since $\tC^{\perp}$ was a triply-even space by assumption, we obtain a tri-orthogonal matrix $\h$ as
\begin{align}
  \h = 
  \begin{pmatrix}
    \h_{1}\\
    \h_{0}
  \end{pmatrix}~.
\end{align}
The resulting quantum code has block size $N-k$ and dimension $k$.
Let $\C_P$ denote the punctured code obtained from $\C$ by removing the first $k$ bits.

Let $\T$ represent the map $\T(\rho) = T\rho T\conj$, where the $T$ gate is defined in eq. (\ref{eq:defT}).
For $p\in [0,1]$, let $\E_{p}$ represent a single-qubit noise channel of strength $p$.
In particular, we will be interested in channels whose action on a density operator $\rho$ describing a single qubit is described by 
\begin{enumerate}
\item the erasure channel $\E_{p}(\rho) = (1-p)\rho + p\kb{\perp}$, where $\perp$ is a symbol denoting erasure, or
\item the dephasing channel $\E_{p}(\rho) = (1-p)\rho + pZ\rho Z$
\footnote{As noted by Bravyi and Haah, if the state is not already in this form, we can force it to be by applying the gates $I$ and $e^{-i\pi/4}SX$ with probability $1/2$ each.}.
\end{enumerate}
We model the noisy $T$ as the composition of ideal gate and the noise channel, i.e. $\E_p \circ \T$.

Distillation proceeds as follows:
\begin{enumerate}
	\item We begin with the state $\eta_0 = \encket{+^{\otimes k}}\encbra{+^{\otimes k}}$, the encoded $k$-fold tensor product of the $\ket{+}$ state over $(N-k)$ physical qubits.
	\item We apply an $(N-k)$-fold tensor product of (noisy) transversal $T$ gates on $\eta_0$ to obtain $\eta_1$
  \[
    \eta_1 = (\E_p\circ\T)^{\otimes (N-k)}(\eta_0)~.
  \]
  \item We apply the Clifford gate guaranteed by lemma \ref{lem:bravyihaah} to $\eta_1$ to obtain $\eta_2$
  \begin{align*}
    \eta_2 = U \eta_1 U\conj &= U\E_p\left(T_{N-k} \encket{+^{\otimes k}}\encbra{+^{\otimes k}} T_{N-k}\conj \right)U\conj\\
    &= \E_p\left(UT_{N-k} \encket{+^{\otimes k}}\encbra{+^{\otimes k}} T_{N-k}\conj U\conj\right)\\
    &= \E_p\left(\encket{A^{\otimes k}}\encbra{A^{\otimes k}}\right)~.
  \end{align*}
  \item \label{item:diverge} At this stage, we utilize Steane's error correction technique \cite{steane1997active}.
  Upon measurement, we obtain a codeword of $\C_P$, and assume that the decoder has been provided the location of the punctures as side information.
  These locations together with the error arising from the noisy channel $\E_{p}$ can be regarded as a composite error channel on the code $\C$.
  The decoder can then proceed to run the decoder of $\C$ on the received codeword where it treats the punctured positions as having suffered erasure noise.
  \item Having deduced the error $E$, we can now perform correction and run the encoder in reverse to obtain $k$ copies of the magic state $\ket{A}$.
\end{enumerate}

We note that our scheme uses Steane error correction at step \ref{item:diverge} rather than measuring the $X$-stabilizers and post-selecting on the all $+1$-outcome.
The polar code is designed to optimize its performance under SC decoding and therefore characterizing the code by its distance can be misleading.
As mentioned in the introduction, the polar code outperforms the Reed-Muller code as measured by its error correction capacity even though both code families are constructed using the same encoding circuits.
Therefore a better metric to study the performance of the polar code is to use its bit error probability under SC decoding.
To compare the two schemes, note that the probability of successfully post-selecting a state scales as $O\left((1-p)^{d}\right)$, and therefore decreases as a function of the block-size (and the number of encoded qubits); the corresponding error rate falls as $O\left(p^{d}\right)$.
On the other hand, decoding guarantees that every state obtained after error correction is part of the codespace whereas the error rate only falls as $O\left(\sqrt{p^{d}}\right)$.
The probability of error in step \ref{item:diverge} is thus upper-bounded by the probability of failure of the classical decoder associated with $\C$.
Understanding the consequences for these trade-offs numerically are directions for future research.

\section{Tri-orthogonal codes from polar codes}
\label{sec:triorthogonalcodesfrompolarcodes}
In this section, we demonstrate the existence of triply-even codes derived from polar codes.
We supplement this with numerics which upper bound the size of the triply-even space and the probability of error associated with the decoding process for block sizes of interest.

Note that according to the prescription above, if we wish to construct a tri-orthogonal code from a code $\C$, we first require its dual $\C^{\perp}$ to be triply-even.
The following series of simple lemmas establish constraints on a decreasing set $\i$ such that $\C(\i)^{\perp}$ is triply-even.

\begin{lemma}
  \label{lem:triplyEvenEquivalance}
  If $\V \subseteq \field_2^{N}$ is triply-even is equivalent to stating that for any two vectors $u,v \in \V$, $u * v \in \V^{\perp}$.
\end{lemma}
\begin{proof}
  For $u,v,w \in \V$, $|u*v*w| = 0 \pmod{2}$ and therefore, any product of the form $u*v$ is orthogonal to vectors in $\V$.
  The other direction also follows trivially.
\end{proof}

\begin{lemma}
\label{lem:bijectionOfStarProduct}
If $\i, \j \subseteq \M_n$ are two sets of monomials, then
\begin{align*}
      \C(\i) * \C(\j) = \C(\i \cdot \j)~.
\end{align*}
\end{lemma}
\begin{proof}
  Let $\i = \{\alpha_1,...,\alpha_{|\i|}\}$ and $\j = \{\beta_1,...,\beta_{|\j|}\}$ denote the elements of these sets.
  Any vectors $v \in \C(\i)$ and $w \in \C(\j)$ can be expressed as $v = \sum_{i=1}^{|\i|} s_i \ev(\alpha_i)$ and $w = \sum_{j=1}^{|\j|} t_j \ev(\beta_j)$ for some constants $s_i, t_j \in \field_2$ for $i \in \i, j \in \j$.
  Since $\ev$ is a bijection and maps the star product of two vectors to the product of monomials, any element of the star product of the two is of the form
  \begin{align*}
    \sum_{i,j} r_{ij} \ev(\alpha_i)*\ev(\beta_j) = \sum_{i,j} r_{ij} \ev(\alpha_i \cdot \beta_j)~,
  \end{align*}
  for some constants $r_{ij} \in \field_2$ for $i \in \i, j \in \j$.
  The desired result follows.
\end{proof}

\begin{lemma}
  Let $\i \subset \M_n$ be decreasing and $\C(\i)$ be the corresponding decreasing monomial code of dimension $K$.
  The space $\C(\i)^{\perp} = \C(\M_n \setminus \comp{\i})$ is tri-orthogonal if and only if
  \[
    (\M_n \setminus \comp{\i}) \cdot (\M_n \setminus \comp{\i}) \subseteq \i~.
  \]
\end{lemma}
\begin{proof}
  From lemma \ref{lem:triplyEvenEquivalance}, the requirement that the code $\C^{\perp} = \C(\M_n \setminus \comp{\i})$ be triply-even is equivalent to
  \begin{align}
    \C(\M_n \setminus \comp{\i}) * \C(\M_n \setminus \comp{\i}) \subseteq \C(\i)~.
  \end{align}

  From lemma \ref{lem:bijectionOfStarProduct}, it follows that this is equivalent to
  \begin{align}
  \label{eq:monomialTriplyEvenRequirement}
    (\M_n \setminus \comp{\i}) \cdot (\M_n \setminus \comp{\i}) \subseteq \i~.
  \end{align}
\end{proof}

These results imply the following corollary which is straightforward to check numerically and doing so gives us the size of the dual code $\C^{\perp}$.
\begin{corollary}
\label{cor:check}
  Let $\i \in \M_n$ be a decreasing set of monomials.
  Finding the smallest code $\C(\i)$ such that $\C^{\perp}$ is triply-even corresponds to finding the smallest set $\i$ for which the condition
  \begin{align}
  \label{eq:condition}
    (\M_n \setminus \comp{\i}) \cdot (\M_n \setminus \comp{\i}) \cap (\M_n \setminus \i) = \emptyset~
  \end{align}
  is still true.
\end{corollary}
\begin{proof}
  Assume that $f \in \M_n \setminus \comp{\i}$ but $f \cdot f \notin \i$.
  This implies that $f \cdot f \in \M_n \setminus \i$, resulting in the condition \ref{eq:condition} above.
  Therefore the smallest code $\C(\i)$ such that $\C^{\perp}$ is triply-even is the smallest code for which the condition
  \[
    (\M_n \setminus \comp{\i}) \cdot (\M_n \setminus \comp{\i}) \cap (\M_n \setminus \i) = \emptyset~
  \]
  is still true.
\end{proof}

We wish to design a polar code $\C$ that will be subject to $N$ copies of a noise channel $\sfw:\field_2 \to \Y$ such that its dual $\C^{\perp}$ is triply-even.
To do so, we first compute the Bhattacharyya parameters corresponding to this channel as explained in subsection \ref{subsec:polarcodes}.
We then order the indices $a \in \{0,1\}^{N}$ according to this parameter, or equivalently, we order the monomials that corresponds to each index $a$.
Given a rate $R \in [0,1]$, we construct the polar code $\C$ using the $RN$ best channels chosen according to their performance.
The threshold $\epsilon \in [0,1]$ which is an upper bound on the bit error rate is the maximum Bhattacharyya parameter for the synthetic channels in $\C$.

These synthetic channels are in one-to-one correspondence with monomials as described in subsection \ref{subsec:decreasingmonomialcodes} and so equivalently, this procedure yields a set of monomials $\i$ of size $RN$.
We then use corollary \ref{cor:check} to verify whether the set $\M_n \setminus \comp{\i}$ satisfies the desired product relation.
Our goal is to construct the smallest set $\i$ satisfying this corollary as this will result in the largest code $\C^{\perp}$, which in turn will maximize the rate of the resulting quantum code.

We remark that this method may not be optimal.
It may be entirely possible that the onerous synthetic channels that violate the triply-even condition are not the poorly performing channels.
We leave the task of optimizing this algorithm as an objective for future research.

\subsection{Dimension of triply-even polar codes}
\label{subsec:dim}
Recall that polar codes are channel specific -- they have to be designed for a particular noise channel.
We first discuss the erasure channel as it is the simplest case to study the polar code.
As explained above, we can estimate the dimension of the channel numerically.
Figure \ref{fig:erasure} shows how the rate of the dual code varies as a function of the log of the block size $N$ for erasure channels whose probability of erasure are $0.01$, $0.02$ and $0.05$ respectively.
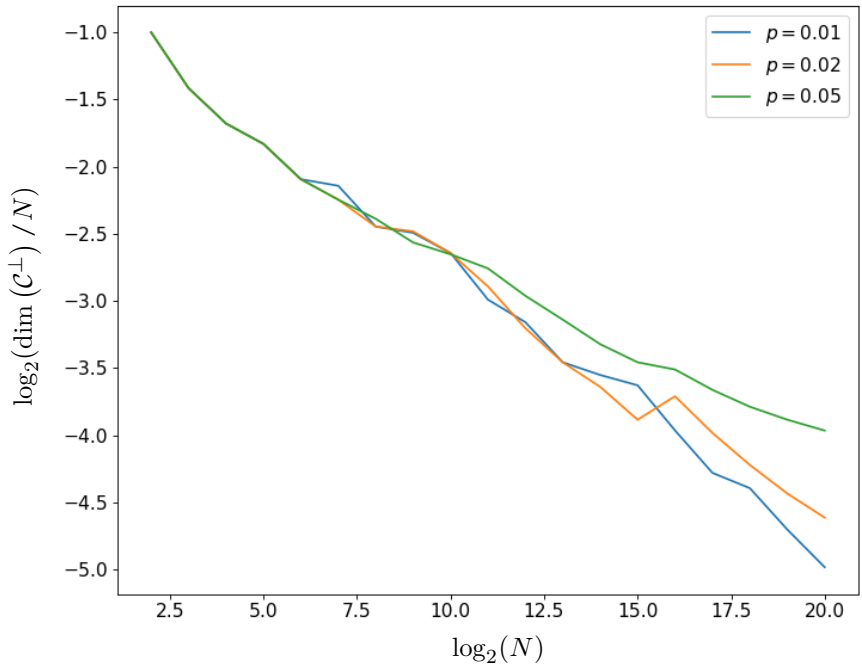
\begin{figure}[h]
\centering
\begin{tikzpicture}
  \node (image) at (0,0) {\includegraphics[scale=0.5]{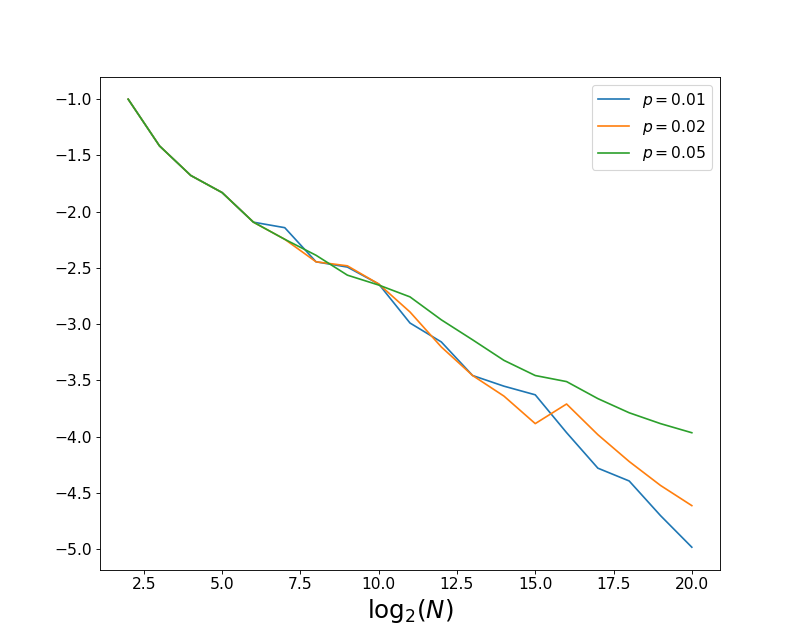}};
  \node [rotate=90] at (-6,0) {$\log_2(\dim{\C^{\perp}}/N)$};
  \fill [white] (-1,-5) rectangle (1.3,-4.3);
  \node at (0.3,-4.7) {$\log_2(N)$};
\end{tikzpicture}
  \caption{Log-log plot of rate $\dim{\C^{\perp}}/N$ of the dual code vs the block-size $N$ for erasure channels with erasure probabilities $0.01$, $0.02$ and $0.05$.
  }
  \label{fig:erasure}
\end{figure}
For an erasure rate of $p=0.01$, the log-log plot and the corresponding line of best fit indicate that the rate of the code $\C^{\perp}$ scales roughly as
\begin{align*}
  \dim{\C^{\perp}}/N = O(N^{-0.2})~,
\end{align*}
or equivalently, that the dimension of the code scales as
\begin{align*}
  \dim{\C^{\perp}} = O(N^{0.8})~.
\end{align*}
We also add that the rates of the codes for larger values of $p$ are better but as we shall soon see, this comes at the cost of error tolerance.

We now proceed to study the binary symmetric channel.
The complexity of designing the polar code as per Arikan's formulation using the Bhattacharyya parameters grows exponentially in the block size.
We instead use Monte Carlo sampling techniques to estimate the Bhattacharyya parameters, the drawback of which is that we require a very large number of samples before we see convergence.
For this reason, the range of block sizes we have explored is not as extensive as that of the erasure channel, and stops at $2^{16} = 65536$ qubits.
Fig. \ref{fig:bsc_dim} is a log-log plot of the rate $\dim{\C^{\perp}}/N$ of the dual code versus the block-size $N$.
\begin{figure}[h]
  \centering
\begin{tikzpicture}
  \node (image) at (0,0) {\includegraphics[scale=0.5]{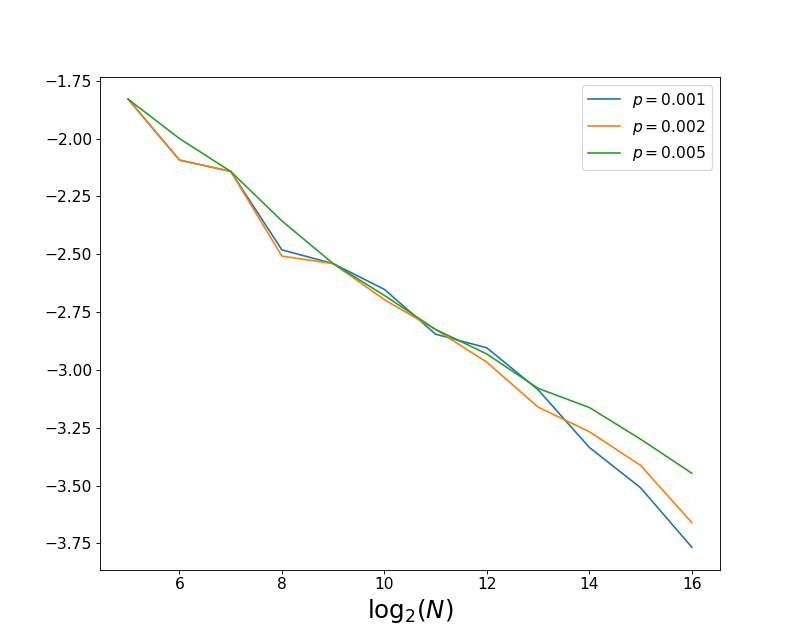}};
  \node [rotate=90] at (-6,0) {$\log_2(\dim{\C^{\perp}}/N)$};
  \fill [white] (-1,-5) rectangle (1.3,-4.3);
  \node at (0.3,-4.7) {$\log_2(N)$};
\end{tikzpicture}
  \caption{Log-log plot of rate $\dim{\C^{\perp}}/N$ of the dual code vs the block-size $N$ for binary symmetric channels at transition probabilities $10^{-3}$, $2\times10^{-3}$ and $5\times10^{-3}$.}
  \label{fig:bsc_dim}
\end{figure}
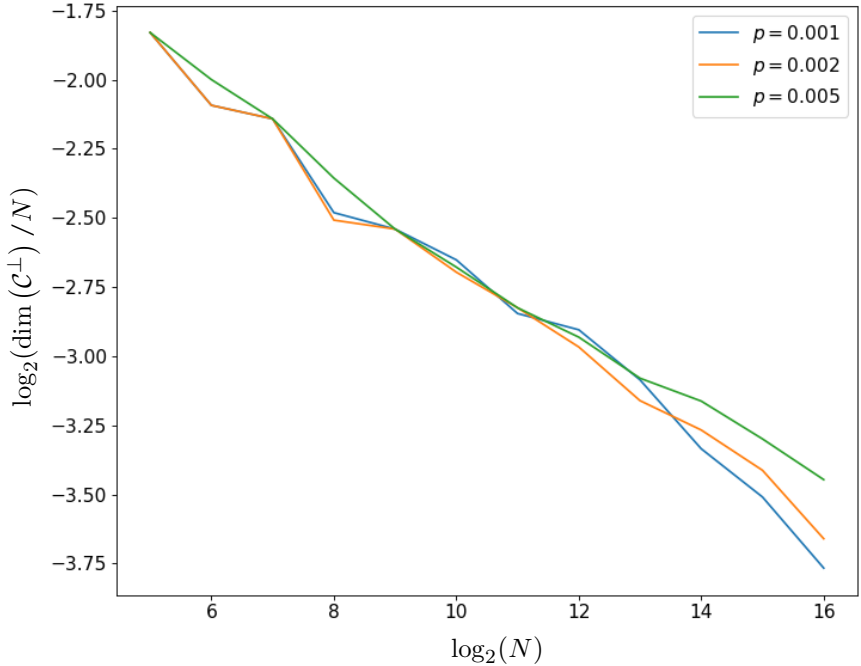
Using a line of best fit, we find that for noise rate $p=0.001$, the dimension of $\C^{\perp}$ versus the block size $N$ scales roughly as $O(N^{0.84})$.

This can be contrasted to the scheme by Hastings and Haah who used punctured Reed-Muller codes to achieve $\C^{\perp} = O(N^{0.91})$.
Note however that this dimension does not tell the whole story as this is only an upper bound on the dimension of the tri-orthogonal quantum code.
More importantly, there is a tradeoff in decoding complexity and performance when using Reed-Muller codes and successive decoders.
The better rate and lower encoding complexity make the polar code more favorable.

The polar code construction is only guaranteed to perform well if the rate of the code $\C$ is below the capacity of the channel it is designed for.
For the erasure channel $\sfw_p$ with erasure probability $p \in [0,1]$, the capacity is $(1-p) \in [0,1]$ and for the binary symmetric channel $\sfv_q$ with transition probability $q \in [0,1/2]$, the capacity is $1-h_2(q) \in [0,1]$, where $h_2\left(\cdot\right)$ is the binary entropy \cite{cover2012elements}.
However, the triply-even constraint implies that the code $\C(\i)$ will eventually encompass the entire space and exceed the maximum size of the polar code for a given noise rate that is guaranteed to perform well.
At that point, this construction is no longer valid as this will mean that the SC decoding can no longer be guaranteed to work.

\subsubsection{Error rates for the binary erasure channel}
\label{subsec:errRateErasure}
In this section, we shall prove that there exist punctures that do not compromise the performance of the code.
The polar code is not designed to maximize the distance, but rather to minimize the probability of decoding error under SC decoding.
A punctured polar code can be seen as a polar code that has suffered erasure errors -- although the punctured bits are not transmitted, the decoder can replace the punctured positions with erasure symbols given the locations of the puncture as side information.
If the noise channel that the code encounters is also an erasure channel, then these two processes -- puncturing and noise -- can be seen as the composition of two erasure channels.
This can be modeled easily thanks to the following observation whose proof is simple and therefore omitted.

\begin{lemma}
\label{lem:erasureJoin}
  The composition of two erasure channels $\sfw_p$ and $\sfw_q$ with erasure probabilities $p$ and $q$ respectively is an erasure channel $\sfw_r$ with erasure probability $r = p + q - pq$.
\end{lemma}

Let $\C$ be a polar code whose dual is triply-even and is designed for the erasure channel $\sfw_r$.
We denote the threshold of $\C$ by $\epsilon \in [0,1]$.
For $p,q \in [0,1]$ such that $p + q -pq = r$, we puncture our polar code $\C$ randomly using an erasure channel $\sfw_{q}$ and it is then subject to erasure noise $\sfw_{p}$.

To show that there exists a good punctured code, we employ a derandomization argument.
\begin{lemma}
  For $\epsilon_0 > \epsilon$, there exists a $\delta \in [0,1]$ such that we can choose with probability $\delta$ a punctured polar code $\tC_{\epsilon_0}$ whose bit error rate is upper bounded by $\epsilon_0$ against erasure noise $\sfw_{p}$ for $p < r$.
\end{lemma}
\begin{proof}
For $q \in [0,1]$, let $\sfw_q$ denote the erasure channel that creates the puncture.
We first apply the channel $\sfw_q^{N}$ which serves to create the punctured code.
Then we apply $\sfw_p^{N}$ which is the noise that the punctured code is subject to.
According to lemma \ref{lem:erasureJoin}, the effective channel that $\C$ is subject to is $\sfw_{r}^{N}$, where $r = p+q -pq$.
Upon transmitting the message $u \in \field_2^N$ across this channel, this could result in erasure patterns that we denote $P,Q \in \field_2^{N}$ where location $a$ has suffered an erasure error if and only if $(P \vee Q)_a = 1$.
Let $\Pr\{u_a \neq \uhat_a| P\vee Q\}$ be the probability that SC decoding results in a decoding error for location $a$.

Since the threshold of $\C$ is $\epsilon$, we can upper bound the bit error rate as
\begin{align}
  \avg_{Q} \avg_{P} \Pr\{u_a \neq \uhat_a| P \vee Q\} \leq \epsilon~,
\end{align}
where $\avg_{Q}$ and $\avg_{P}$ denote the expectation value over random erasure patterns $Q,P \in \field_2^N$.

For $\delta \in [0,1]$, we may now apply a Markov inequality to upperbound the probability of picking a `bad' erasure pattern $Q$
\begin{align}
  \Pr_{Q}\left\{ \avg_{P} \Pr\{u_a \neq \uhat_a| P\vee Q\} \geq \epsilon_0 \right\} \leq \delta~,
\end{align}
where we define the threshold $\delta = \epsilon/\epsilon_0$.
Hence for $\epsilon_0 > \epsilon$, there exist punctures $P$ such that the probability of failure for any bit is upper bounded by $\epsilon_0$.
\end{proof}

Let $\C$ be the smallest polar code such that $\C^{\perp}$ is triply-even, and let $\epsilon$ be its threshold.
In fig. \ref{fig:bec_err} below, we plot the corresponding log-likelihood ratios (LLRs) vs the log of the block size.
In this case, the LLR corresponding to $\C$ is $\log_2((1-\epsilon)/\epsilon)$.

We find that the probability of decoding error drops significantly as the block-size increases.
This suggests that for block sizes of interest, it might suffice to increase the block size of the polar code rather than perform concatenation as in the Bravyi-Haah distillation scheme.
We can see that the code designed for $p = 0.01$ has an error rate several orders of magnitude better for block sizes of interest.
Thus for a minor tradeoff in the dimension with respect to codes designed for $p=0.02$ and $p=0.05$, the code designed for $p=0.01$ sees a tremendous benefit with respect to the error rate.
We can see that for a block size of $2^{18} = 262,144$ at erasure error rate $0.01$, we can achieve a bit error rate of roughly $2^{-90} \approx 8 \times 10^{-28}$.

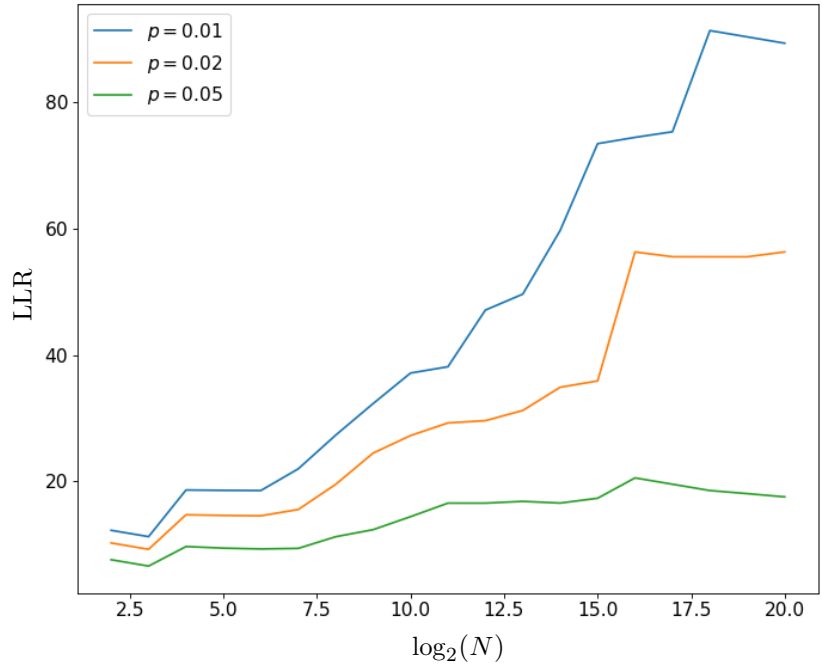
\begin{figure}[h]
  \centering
  \begin{tikzpicture}
      \node (image) at (0,0) {\includegraphics[scale=0.5]{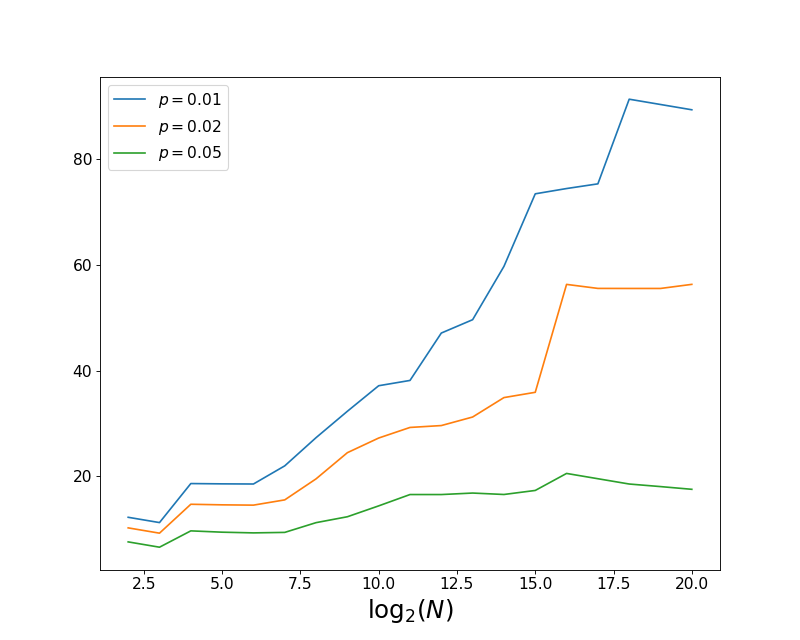}};
      \node [rotate=90] at (-5.5,0) {$\operatorname{LLR}$};
      \fill [white] (-1,-5) rectangle (1.3,-4.3);
      \node at (0.3,-4.7) {$\log_2(N)$};
  \end{tikzpicture}
  \caption{Variation of the log-likelihood ratios (LLRs) vs the log of the block-size for erasure probabilities $p = 0.01$, $p=0.02$ and $p=0.05$.}
  \label{fig:bec_err}
\end{figure}

\subsubsection{Error rates for the binary symmetric channel}
\label{subsec:errRateBSCp}
In the event that the noise channel that the quantum code is subject to is a dephasing channel, then we must deal with a composition of a binary erasure channel (which creates the puncture) and a binary symmetric channel (which is the noise).
Before doing so, it will be convenient to introduce the notion of channel degradability as in \cite{bardet2016algebraic}.
For some sets of alphabets $\X$, $\Y$ and $\Z$, we say that a channel $\sfu: \X \to \Y$ can be \emph{degraded} to a channel $\sfut:\X \to \Z$ if there exists a channel $\sfq:\Y \to \Z$ such that
\begin{align*}
  \sfut(z|x) = \sum_{y\in \Y} \sfq(z|y)\sfu(y|x)~.
\end{align*}
If such a channel $\sfq$ exists, this is denoted $\sfut \cleq \sfu$.
Equivalently, we can say that $\sfu$ can be degraded to $\sfu'$ if there exists a channel $\sfq$ that uses $\sfu$ as a sub-routine to simulate $\sfut$.

We use the following lemma, each of whose claims has been proved earlier.
\begin{lemma}
\label{lem:orderDegrade}
For channels $\sfu:\X \to \Y$ and $\sfut:\X \to \Y$,
\begin{enumerate}
  \item (\cite{richardson2008modern} (p207)) the reliability of $\sfu$ is at least as good as that of $\sfut$, i.e.
  \[
    \sfut \cleq \sfu \implies \B(\sfu) \leq \B(\sfut)~;
  \]
  \item (\cite{kobara2009code} (lemma 4.7)) for all block sizes $N$, the synthetic channels inherit this relation, i.e. for $a \in [N]$,
  \[
    \sfut \cleq \sfu \implies \sfut^{(a)} \cleq \sfu^{(a)}~.
  \]
\end{enumerate}
\end{lemma}

In other words, if $\sfu$ can be degraded to $\sfu'$, then $\sfu$ is at least as reliable as $\sfu'$.
Using the notion of degradability, we can simplify the analysis for the binary symmetric channel by degrading the erasure channel using the following lemmas.

\begin{lemma}
\label{lem:degradeErasureToBSCp}
Let $p \in [0,1]$ be some erasure probability.
\begin{enumerate}
  \item The erasure channel $\sfw_p$ of erasure probability $p$ is degradable to a binary symmetric channel of transition probability $p/2$, i.e.
  \begin{align*}
    \sfv_{p/2} \cleq \sfw_{p}~.
  \end{align*}
  \item The composition of binary symmetric channels $\sfv_{a}$ and $\sfv_{b}$ is a binary symmetric channel $\sfv_r$ where $r = p + q -2pq$.
\end{enumerate}
\end{lemma}
\begin{proof}
We shall prove each of these statements in turn.
\begin{enumerate}
  \item Consider transmitting a bit across the binary erasure channel $\sfw_p$ of erasure probability $p$.
  With probability $p$, the bit is erased and we may replace it by a random symbol $0$ or $1$.
  Hence with probability $p$, the effective channel $\Z$ is a binary symmetric channel with probability $1/2$ and is the ideal channel otherwise.
  For $x \in \field_2$, we have
  \begin{align*}
    \Z(x) &= (1-p)\sfv_{0}(x) + p\sfv_{1/2}(x)\\
          &= (1-p)[x] + \frac{p}{2}[x] + \frac{p}{2}[\bar{x}]\\
          &= \left(1-\frac{p}{2}\right)[x] + \left( \frac{p}{2}\right)[\bar{x}]~.
  \end{align*}
  Thus $\Z$ is equivalent to the binary symmetric channel $\sfv_{r}$, where $r = p/2$.

  \item For $x \in \field_2$, the action of the composition $\sfv_{a}\circ\sfv_{b}$ is 
  \begin{align*}
    \left(\sfv_{a} \circ \sfv_{b}\right)(x) &= \sfv_{a}\left((1-b)[x] + b[\bar{x}] \right)\\
                                            &= \left((1-a)(1-b) + ab\right)[x] + \left(a(1-b) + (1-a)b \right)[\bar{x}]\\
                                            &= \left(1-a-b+2ab\right)[x] + \left(a + b -2ab \right)[\bar{x}]~.
  \end{align*}
\end{enumerate}
\end{proof}

Of course, estimating the probability of failure by degrading the erasure channel to a binary symmetric channel only yields a lower bound.
This is because we know the erasure locations when we transmit across the BEC but this is no longer true when we code for the binary symmetric channel.

Let $\C$ be a polar code whose dual is triply-even and is designed for the composition of the erasure channel $\sfw_q$ and binary symmetric channel $\sfv_p$.
We denote the threshold of $\C$ by $\epsilon \in [0,1]$.
For $p,q \in [0,1]$ such that $p + q/2 -pq = r$, we puncture our polar code $\C$ randomly using an erasure channel $\sfw_{q}$ and it is then subject to the binary symmetric channel $\sfv_{p}$.

As was done for the erasure channel, to show that there exists a good punctured code, we employ a derandomization argument.
\begin{lemma}
  For $\epsilon_0 > \epsilon$, there exists a $\delta \in [0,1]$ such that we can choose with probability $\delta$ a punctured polar code $\tC_{\epsilon_0}$ whose bit error rate is upper bounded by $\epsilon_0$ against erasure noise $\sfw_{p}$ for $p < r$.
\end{lemma}
\begin{proof}
  The code $\C$ is designed for the composition $\Z := \sfw_q \circ \sfv_p$ of the erasure channel $\sfw_q$ and the binary symmetric channel $\sfv_p$.
  The binary erasure channel $\sfw_q$ can be degraded to a binary symmetric channel $\sfv_{q/2}$ according to lemma \ref{lem:degradeErasureToBSCp}.
  Noting that the composition of the channels $\sfv_{p}$ and $\sfv_{q/2}$ is yet again a binary symmetric channel $\sfv_{r}$, where $r = p + q/2 - pq$, we have $\sfv_{r} \cleq \Z$.
  It follows then from lemma \ref{lem:orderDegrade} that the synthetic channels obtained by designing polar codes for $\Z$ can be degraded to those synthetic channels obtained by designing polar codes for $\sfv_{r}$.

  The proof now follows similarly to the case of the erasure channel we have dealt with.

  Transmitting the message $u\in \field_2^N$ across $\sfv_{r}$ could result in error patterns that we denote $P,Q \in \field_2^{N}$ where location $a$ has suffered an erasure error if and only if $(P + Q)_a = 1$.
  Let $\Pr\{u_a \neq \uhat_a| P + Q\}$ be the probability that SC decoding results in a decoding error for location $a$.

  Since the threshold of $\C$ is $\epsilon$, we can upper bound the bit error rate as
  \begin{align}
    \avg_{Q} \avg_{P} \Pr\{u_a \neq \uhat_a| P + Q\} \leq \epsilon~,
  \end{align}
  where $\avg_{Q}$ and $\avg_{P}$ denote the expectation value over random error patterns $Q,P$.

  For $\delta \in [0,1]$, we may now apply a Markov inequality to upperbound the probability of picking a `bad' erasure pattern $Q$
  \begin{align}
    \Pr_{Q}\left\{ \avg_{P} \Pr\{u_a \neq \uhat_a| P + Q\} \geq \epsilon_0 \right\} \leq \delta~,
  \end{align}
  where we define the threshold $\delta = \epsilon/\epsilon_0$.
  Hence for $\epsilon_0 > \epsilon$, there exist punctures $P$ such that the bit error rate for any bit is upper bounded by $\epsilon_0$.
\end{proof}

Let $\C$ be the smallest polar code such that $\C^{\perp}$ is triply-even, and let $\epsilon$ be its threshold.
In figure \ref{fig:err_bsc} below, we plot $-\log_{2}(\epsilon)$ versus the log of the block size of the code $\C$ for noise rates $0.001$, $0.002$ and $0.005$.
For the noise rates we are studying, we expect the performance of the polar codes designed for the binary symmetric channel to be worse than those designed for the binary erasure channel.
We can see that for block sizes of $2^{16} = 65,536$ for a noise rate of $0.001$, we can achieve a bit error rate of roughly $2^{-46} \approx 7 \times 10^{-15}$.

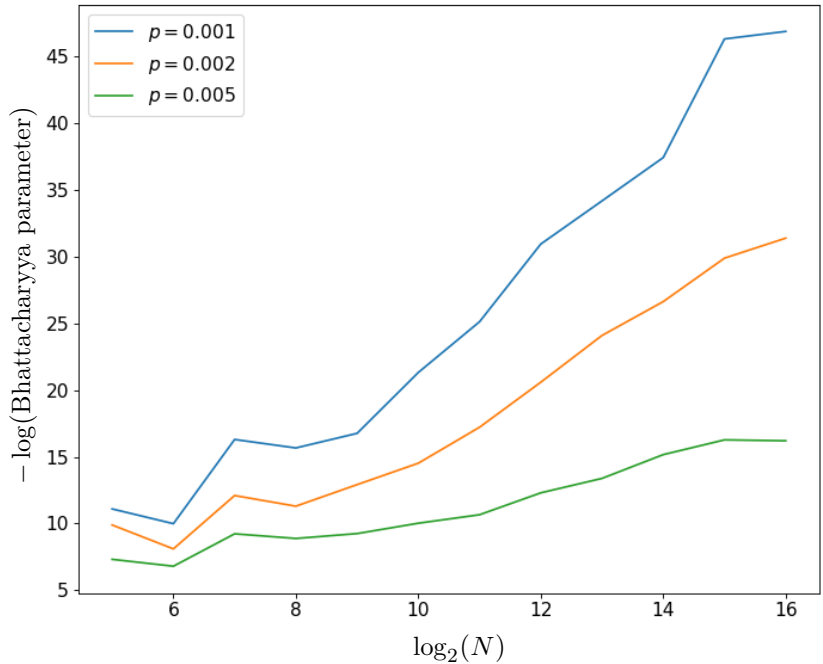
\begin{figure}[h]
  \centering
  \begin{tikzpicture}
      \node (image) at (0,0) {\includegraphics[scale=0.5]{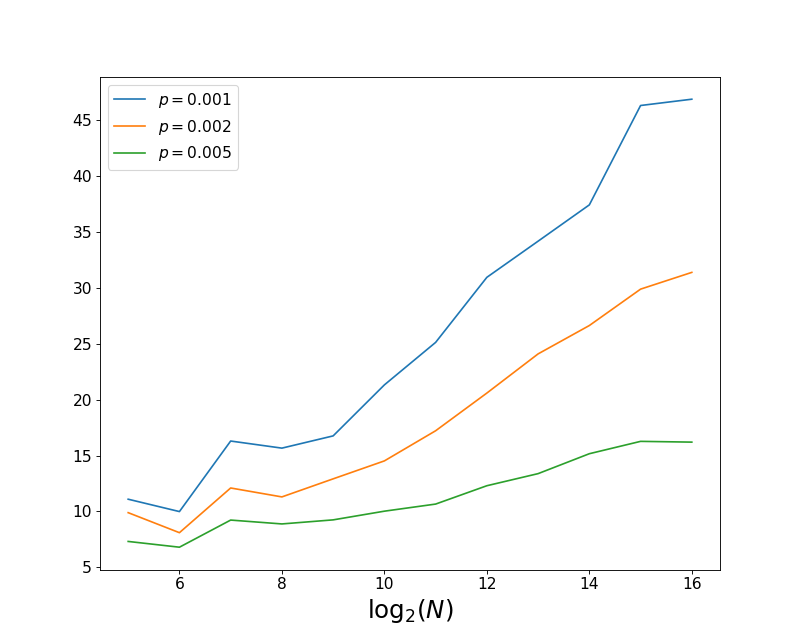}};
      \node [rotate=90] at (-5.5,0) {$-\log($Bhattacharyya parameter$)$};
      \fill [white] (-1,-5) rectangle (1.3,-4.3);
      \node at (0.3,-4.7) {$\log_2(N)$};
  \end{tikzpicture}
  \caption{Log-log plot of the Bhattacharyya parameters vs the block-size for transition probabilities $p = 0.001$, $p=0.002$ and $p=0.005$.}
  \label{fig:err_bsc}
\end{figure}

\section{Conclusions}
We have demonstrated how to use polar codes to construct tri-orthogonal codes, which in turn can be used for magic state distillation using the framework of Bravyi and Haah.
We introduced new theoretical tools from recent advances in classical coding theory which we hope will help in the broader study of quantum error correcting codes.
In addition, using polar codes allows us to use the low complexity successive cancellation decoder which functions in $O(N\log N)$ time for a polar code of block size $N$.
By puncturing these codes, we can achieve a growing dimension for the number of encoded logical qubits.
This decoder is known to have a better error correction capacity than decoders for Reed-Muller codes.
The dimension of these triply-even codes is shown to scale like $O(N^{0.8})$ for the binary erasure channel at noise rate $0.01$ and $O(N^{0.84})$ for the binary symmetric channel at noise rate $0.001$.
Finally it was shown that we can upper bound the probability of failure of decoding a bit using the Bhattacharyya parameters for the corresponding classical codes.
The error probability for the triply-even codes can drop to roughly $8\times10^{-28}$ for the erasure channel at a block size of $2^{20} = 262,144$ and $7 \times 10^{-15}$ for the dephasing channel at a block size of $2^{16} = 65,536$.

\textbf{Future directions:} It is well known that polar codes equipped with the successive cancellation decoder require large block-lengths to be effective.
In order to address this problem, Tal and Vardy have devised a successive cancellation list decoder \cite{tal2015list}.
Together with some pre-encoding with cyclic repetition codes, the performance of these codes is significantly improved even for smaller block sizes.
For practical implementations, it would be interesting to see how these ideas can be used to improve the performance of polar codes for magic state distillation. It would also be interesting to perform numerical simulations to see how the performance of punctured polar codes directly compares to punctured Reed-Muller codes.

\section{Acknowledgements}
We would like to thank David Poulin for many helpful discussions.
In particular, we would like to thank him for drawing our attention to the argument in section \ref{sec:distillationprocedure} contrasting error correction and post-selection.
We would also like to thank Colin Trout for detailed feedback on an earlier draft of this work, and Pavithran Iyer and Maxime Tremblay for many helpful comments.
A.K. would like to thank the MITACS organization for the Globalink award which facilitated his visit to Inria, Paris and Inria, Paris for their hospitality during his visit.
A.K. also acknowledges support from the Fonds de Recherche du Qu\'ebec - Nature et Technologies (FRQNT) via the B2X scholarship for doctoral candidates.
Computations were made on the supercomputers managed by Calcul Qu\'ebec and Compute Canada.
The operation of these supercomputers is funded by the Canada Foundation for Innovation (CFI),
the Minist\`ere de l'Economie, de la Science et de l'Innovation du Qu\'ebec (MESI) and the FRQNT.
JPT acknowledges the support of the European Union and the French Agence Nationale de la Recherche through the QCDA project.

\bibliographystyle{unsrtabbrev}
\bibliography{references}
\end{document}